\documentclass[wcp]{jmlr}
\usepackage{amsmath,amssymb,graphicx,url}
\jmlrvolume{204}
\jmlryear{2023}
\jmlrworkshop{Conformal and Probabilistic Prediction with Applications}
\jmlrproceedings{PMLR}{Proceedings of Machine Learning Research}

\usepackage{aligned-overset}
\usepackage{algorithm}%
\usepackage{xcolor}
\usepackage{algorithm2e}

\title{Anomalous Edge Detection in Edge Exchangeable Social Network Models}

\author{\Name{Rui Luo}\Email{ruiluo@cityu.edu.hk} \\
       \addr{City University of Hong Kong, Kowloon Tong, Hong Kong} \\
       \addr{Cornell University, Ithaca, USA}\\
       \Name{Buddhika Nettasinghe}\Email{buddhika-nettasinghe@uiowa.edu} \\
       \addr{University of Iowa, Iowa City, USA}\\
       \Name{Vikram Krishnamurthy}\Email{vikramk@cornell.edu} \\
       \addr{Cornell University, Ithaca, USA}}

\editor{Harris Papadopoulos, Khuong An Nguyen, Henrik Boström and Lars Carlsson}

\begin{document}

\maketitle

\begin{abstract}%
This paper studies detecting anomalous edges in directed graphs that model social networks. We exploit edge exchangeability as a criterion for distinguishing anomalous edges from normal edges. Then we present an anomaly detector based on conformal prediction theory; this detector has a guaranteed upper bound for false positive rate. In numerical experiments, we show that the proposed algorithm achieves superior performance to baseline methods.
\end{abstract}


\section{Introduction}
\label{ADND-sec:introduction}
The rapid growth of online social networks raises security concerns, including coordinated spam attacks \citep{wagner2012social}, cyber-bullies \citep{bindu2016mining}, and dissemination of misinformation \citep{nguyen2012containment}. The problem we address is: how do we spot such activities in a network? Specifically, we represent them as anomalous edges in a graph, which show significant deviations from normal connection patterns in the network. In this regard, we aim to develop algorithms to detect anomalous edges in social networks.

A number of novel graph-based approaches have been proposed for anomaly detection in networks, which account for the correlations of objects in networks \citep{akoglu2015graph}. Among the approaches focusing on detecting anomalous edges, some share a similar perspective with link prediction algorithms, using the proximity of adjacent nodes \citep{yu2018netwalk, akoglu2010oddball} or factorization of the adjacency matrix \citep{tong2011non, chakrabarti2004autopart} as evidence of abnormal connections. Other approaches analyze anomalous edges by aggregating them into graph snapshots \citep{shin2018patterns, eswaran2018spotlight}. One issue is that these graph- or snapshot-based approaches require high memory and often prohibitive time lag as well. %

More recently, \citet{ranshous2015anomaly, eswaran2018sedanspot, chang2021f} have proposed to detect individually surprising edges in the edge stream which have superior memory and time efficiency. However, due to the non-i.i.d. nature of graph-structured data, they fail to obtain provable guarantees for these algorithms. \citet{bhatia2020midas} develops a hypothesis testing framework and has shown guarantees on the false positive rate. However, it focuses on detecting micro-clusters of edges.

In this paper, we use edge exchangeability as a criterion for detecting anomalous edges in networks.
We then propose an algorithm which guarantees the false positive rate by combining edge exchangeability with conformal prediction. Conformal prediction generates set-valued predictions (the set covers the ground truth with high probability) and has been used successfully in several machine learning applications, including image classification \citep{shafer2008tutorial} and anomalous trajectory detection \citep{laxhammar2010conformal}. Specifically, we fit an edge exchangeable model to the network data. Then, we construct an algorithm (Conformal detector) based on conformal prediction to detect anomalous edges which break exchangeability. 

\noindent
{\bf Main Results and Organization: }

\noindent(1) Section \ref{ADND-sec: EEM} reviews the edge exchangeable models and justifies that edge exchangeability captures the nature of many real-world networks.

\noindent(2) Section \ref{ADND-subsec: ADND} presents the edge exchangeable model and related inference approach according to the asymmetric Dirichlet network distribution (ADND) \citep{williamson2016nonparametric}. Section \ref{ADND-subsec: detect anomaly} gives the definition of anomalous edges regarding edge exchangeability. Section \ref{ADND-subsec: detect anomaly} proposes the anomalous edge detection algorithm (Conformal detector) based on conformal prediction.

\noindent(3) Section \ref{ADND-sec: numerical} presents numerical results that demonstrate the performance of the proposed Conformal detector in anomalous edge detection using transaction data in Ethereum blockchain.

\noindent(4) Finally, the appendix proves that the proposed Conformal detector guarantees the upper bound of false positive rate (detecting normal edges as an anomaly).

\subsection{Related Work}
Previous related works can be considered under three categories:

\noindent(1) {\it Anomalous edge detection in graphs:}  
Existing anomalous edge detection methods focus on the following three aspects. 
\begin{list}{\labelitemi}{\leftmargin=1em}
\item Graph compression: \citet{chakrabarti2004autopart} encoded an input graph by alternatively arranging the nodes into disjoint partitions and splitting the partitions; then spots the anomalous edges whose removal reduces the total encoding cost significantly. \citet{tong2011non} factorized
the adjacency matrix of the input graph and flag edges with high reconstruction error as anomalies. \citet{aggarwal2011outlier} constructed a structural connectivity model of the input graph to define anomalous edges.

\item Node and subgraph similarity: \citet{yu2018netwalk} learned a dynamic node embedding which minimizes the pairwise distance of node embeddings along each walk and computes the abnormality of edges according to the Hadamard product of node embeddings. \citet{akoglu2010oddball} extracted egonet-based (1-hop neighborhood around a node) features and spots anomalous edges according to abnormal pairwise correlation patterns. \citet{shin2018patterns} and \citet{eswaran2018sedanspot} aggregated edges into graph snapshots and detect dense subgraphs with anomalous nodes and edges.

\item Edge frequency: \citet{ranshous2016scalable} proposed an edge probability function which scores each edge based on occurrence frequency, common neighbors and popularity of two endpoints. \citet{eswaran2018spotlight} identified edge anomalies as bursty occurrence which connect parts of the graph which are sparsely connected. \citet{bhatia2020midas} proposed a probabilistic framework to detect microcluster anomalies (suddenly arriving of groups of suspiciously similar edges). \citet{chang2021f} used a frequency-factorization technique to model the distributions of edge frequencies and flags edges with low likelihood of the observed frequency as anomaly. 
\end{list}

\vspace{0.1cm}
\noindent(2) {\it Exchangeable graph models:} 
Vertex exchangeable random graph models assume that the distribution of a graph is invariant to the labeling of its vertices \citep{wang1987stochastic, hoff2002latent, kemp2006learning, miller2009nonparametric}. An extensive list of vertex exchangeable random graph models can be found in \citep{lloyd2012random}. A limitation of vertex exchangeable random graph models is that they yield dense graphs \citep{lloyd2012random}. Edge exchangeable random graph models, on the other hand, are capable to produce sparse graphs and small-world behavior of real-world networks \citep{janson2018edge}. Several edge exchangeable models \citep{cai2016edge,crane2018edge,williamson2016nonparametric} have been proposed for modeling networks as exchangeable sequences of edges. These models assume that the distribution of a graph's edge sequence is invariant to the order of the edges \citep{cai2016edge}. 
In addition, \citet{caron2017sparse} proposed representing a network as an exchangeable point process and using completely random measures to construct sparse graphs.

\vspace{0.1cm}
\noindent(3) {\it Conformal prediction:}
Conformal prediction algorithms provide valid measures of confidence for individual predictions made by machine learning algorithms \citep{gammerman2007hedging}. It is based on the assumption that the samples are probabilistically exchangeable \citep{shafer2008tutorial}, which is slightly weaker than i.i.d. assumption. Different machine learning algorithms have previously been proposed and investigated for conformal prediction, e.g., the k-nearest neighbours algorithm \citep{vovk2005algorithmic}, SVM \citep{vovk2005algorithmic}, and neural networks \citep{papadopoulos2007conformal}. The computational efficiency of conformal prediction is considerably improved using a modification known as inductive conformal prediction \citep{papadopoulos2008inductive}, which avoids training the algorithm one time for each sample in the dataset. Conformal prediction has been used for testing exchangeability of a sequence of data \citep{vovk2003testing, fedorova2012plug, volkhonskiy2017inductive, cai2020real}. In addition, \citet{laxhammar2010conformal} applied conformal prediction in detecting anomalous trajectories in surveilance applications; \citet{ishimtsev2017conformal} proposed an anomaly detection algorithm for univariate time series by combinning k-nearest neighbors algorithm with the conformal prediction framework. Recently, \citep{https://doi.org/10.48550/arxiv.2210.02271} extended conformal prediction to the setting of a Hidden Markov Model, which greatly relaxes the i.i.d. assumption.

\begin{table} \label{table:symbols}
  \caption{Glossary of Symbols Used in This Article}
  \centering 
    \begin{tabular}{ cc } %
    \hline
    Symbols  & Description \\ %
    \hline
    $G$ & the directed graph representing the network \\
    $V$ & the set of nodes \\
    $E$ & the set of edges \\
    $N$ & the number of edges in $E$ \\
    $e_n$ & the edge in $E, n=1,\cdots, N$ \\
    $e_{N+1}$ & the new edge to be tested \\
    $u_n$ & the sender node of edge $e_n$ \\
    $v_n$ & the receiver node of edge $e_n$ \\
    $\alpha_n$ & the non-conformity score of edge $e_n$ \\
    $p_n$ & the $p$-value of edge $e_n$ being normal \\
    $\epsilon$ & the significance level \\
    $\textrm{Anom}_{n}^{\epsilon}$ & indicator of whether edge $e_n$ is anomalous or not at the significance level $\epsilon$ \\
    $A$ & The sender-specific topic distribution \\
    $B$ & The receiver-specific topic distribution \\
    $H$ & The base distribution for $A$ and $B$ \\
    $\tau$ & concentration parameters of $A$ and $B$ \\
    $\gamma$ & concentration parameters of $H$ \\
    $ \overset{d}{=} $ & equality in distribution \\

    \hline
    \end{tabular}
  \end{table}

\section{Edge Exchangeable Social Network Model}
\label{ADND-sec: EEM}
In this section, we first introduce the edge exchangeable models and their connections with the vertex exchangeable models. We then justify how these models capture some important properties of real world social networks, including sparsity and power-law degree distribution. We further illustrate via several real world examples that edge exchangeability is a useful assumption in modeling social networks, which guides the definition of anomalous edges in Section \ref{ADND-subsec: detect anomaly} and the anomaly detection algorithm in Section \ref{ADND-subsec: detect anomaly}.

\subsection{Edge Exchangeable Model}
\label{ADND-subsec: SBM intro}

Consider a directed graph $G=(V,E)$, where $V=\{v_1,\cdots,v_M\}$ is the node set with $M$ nodes and $E=\{e_1,\cdots,e_N\}$ is the edge set with $N$ edges. A sequence of random variables is (finitely) exchangeable if their joint distribution is invariant under permutations \citep{niepert2014exchangeable}. Exchangeability reflects the assumption that the variables do not depend on their indices although they may be dependent among themselves. Exchangeability\footnote{The simplest type of exchangeable process is sampling without replacement of an iid process \citep{orbanz2010bayesian}.} is a weaker assumption than iid; iid variables are automatically exchangeable \citep{orbanz2010bayesian}. 

This paper deals with edge exchangeable networks. We now briefly discuss vertex exchangeable networks. In graphs, the \emph{vertex exchangeability}, defined as follows, means that the distribution of the graph is invariant to relabelings of the vertices. 

\begin{definition}[Vertex Exchangeability]
Consider a graph $G=(V,E)$, where $V=\{1,\cdots,M\}$ and $E=\{e_1,\cdots,e_N\} \subseteq V\times V$. Such networks can be expressed as an adjacency matrix $\textbf{y}=(y_{ij})_{i,j\in V}$ with
\begin{equation}
    y_{ij} = \begin{cases}
     1, & (i,j)\in E, \\
     0, & \textrm{otherwise}.
    \end{cases}
\end{equation}
Let permutation $\pi:V \rightarrow V$. Relabeling vertices according to $\pi$ leads to the adjacency matrix $\textbf{y}^{\pi} = (y_{\pi(i)\pi(j)})_{i,j\in V}$. 
Then, $G$ is vertex exchangeable if for every permutation $\pi$, $\textbf{y}^{\pi} \overset{d}{=} \textbf{y}$. Here, $\overset{d}{=}$ denotes equality in distribution, that is,
\begin{equation}
    P(Y = \textbf{y} ) = P(Y = \textbf{y}^{\pi}), 
\end{equation}
where $Y=(Y_{ij})_{i,j\in V}$ is a random matrix defined on $ \{0,1\}^{V \times V}$. 
\end{definition}

The assumption of vertex exchangeability underlies many well-known network models, including the stochastic block model (SBM) \citep{wang1987stochastic}, latent space model \citep{hoff2002latent}, infinite relational model (IRM) \citep{kemp2006learning}, latent feature relational model (LFRM) \citep{miller2009nonparametric} and many others. However, they fall under the framework of Aldous-Hoover representation theorem \citep{aldous1981representations,hoover1979relations}, which guarantees that these graphs are either empty or dense with probability one \citep{orbanz2014bayesian}. Such models cannot produce sparsity in real-world graphs.

In comparison to vertex exchangeable networks, edge exchangeable networks \citep{cai2016edge,crane2018edge,williamson2016nonparametric} considered in this paper capture important properties of large-scale social networks such as sparsity, community structure and power law degree distributions.  Roughly speaking, for an edge exchangeable graph, the joint distribution of the edge sequence is invariant to the order of the edges. We define this formally as follows:

\begin{definition}[Edge Exchangeability]
Consider a graph $G=(V,E)$, where $V=\{1,\cdots,M\}$ and $E=\{e_1,\cdots,e_N\} \subseteq V\times V$. Each edge is represented as a sender-receiver pair $(u,v) \in V\times V$. Let permutation $\pi:\{1,\cdots,N\} \rightarrow \{1,\cdots,N\}$. 
Then, $G$ is edge exchangeable if for every permutation $\pi$, the edge sequence $(e_1,\cdots,e_N)$ is exchangeable, i.e.,
\begin{equation}
  (e_1,e_2,\cdots,e_{N}) \overset{d}{=} (e_{\pi(1)},e_{\pi(2)},\cdots,e_{\pi(N)}),
\end{equation}
for every permutation $\pi$. Here, $\overset{d}{=}$ denotes equality in distribution, that is,
\begin{equation}
    P(e_n=(u_n, v_n),  \scriptstyle{n=1,\cdots,N} \displaystyle) = P(e_n=(u_{\pi(n)}, v_{\pi(n)}), \scriptstyle n=1,\cdots,N \displaystyle), 
\end{equation}
for every permutation $\pi$.    
\end{definition}

\subsection{Why Edge Exchangeability}
\label{ADND-subsec: why edge exchange}
In this subsection, we demonstrate how edge exchangeability accounts for the complex behavior of social networks. First, in situations where the focus is on the interactions between social network members, e.g., communications, collaborations, and relationships, it is reasonable to represent edges as the statistical units of the network \citep{crane2018probabilistic}. Second, edge exchangeability assumes that the observed edges are a representative sample of the population of all edges \citep{crane2018edge}. In the phone call database example \citep{crane2018edge}, each phone call has the same probability to be sampled, leading to callers who make a large number of phone calls appearing in the database with a high frequency.
In contrast, vertex exchangeability assumes that sampled vertices are representative of the population of all vertices, which neglects the different frequencies of callers in making phone calls. Another advantage is that edge exchangeable models suffice to model interactions driven by edge growth. 

Although exchangeability often fails to hold for time series data due to the dependence structure in the data, exchangeability holds for networks where %
the distribution of edges is invariant under the addition of new edges\footnote{In other words, new edges are sampled without replacement from the same distribution of the existing edges.}. 
We justify edge exchangeability in social networks via the following examples. In the first two examples, edges are sampled without replacement from a fixed distribution, and the permutation of the edges will not affect their joint distribution.

\vspace{0.2in}

\noindent
{\bf Example 1. Online discussion thread:} A discussion thread gathers people interested in the same topics and their messages on such topics (e.g., sports, politics). In the thread, people have different roles such as opinion leaders, followers, and moderators, and they each have a regular pattern of participating in the discussion. If one person $u$'s post is commented on by another person $v$, $u$ will have a directed edge to $v$, forming an ordered pair $(u,v)$. On the online thread, individuals develop a consistent pattern of interacting with each other. Following the definition in \citet{crane2018edge}, we denote $\mathcal{P}$ as a set of participants in the discussion thread. %
Now we define a distribution $f=(f_e)_{e\in \textrm{fin}_2(\mathcal{P})}$ on $\textrm{fin}_2(\mathcal{P})$, where $\textrm{fin}_2(\mathcal{P})$ is the set of all ordered pairs of $\mathcal{P}$,
\begin{equation} \label{ADND-eq:crane distribution}
    P((u,v)=e|f) = f_e, \quad  e\in \textrm{fin}_2(\mathcal{P}),
\end{equation}
where
\begin{equation}
    (f_e)_{e\in \textrm{fin}_2(\mathcal{P})}: f_e\geq 0 \quad \textrm{and} \quad \sum_{e\in \textrm{fin}_2(\mathcal{P})} f_e = 1.
\end{equation}

Let the edges $(u_1, v_1), \cdots, (u_n, v_n)$ be independent, identically distributed random multisets of size 2 drawn from (\ref{ADND-eq:crane distribution}). This generating mechanism results in an exchangeable edge sequence.

\vspace{0.2in}
In addition to the above examples, \citep{crane2018edge} and \citep{williamson2016nonparametric} have verified that the edge exchangeable model outperforms vertex exchangeable model counterparts (IRM, SBM) in predictive likelihood on several real world datasets.

On the other hand, anomalous edges will likely lie an abnormal distance from other edges in the population and therefore causing the edge sequence nonexchangeable. Regarding Example 3 below, a situation where the new edges are anomalous and lead to nonexchangeable edge sequence is as follows.

\vspace{0.2in}
\noindent
{\bf Example 2. (Continued from Example 1) Nonexchangeable online discussion thread:} A hacker $h$ wants to attack the thread by spreading some fake news. He either joins several (trending) existing discussions by commenting about fake news in a post or starts a discussion by posting fake news. The first action results in him having in-edges with people in the existing discussions, e.g., $(u, h)$; the second action results in him having out-edges with people deluded by the fake news, e.g., $(h, v)$. Either action signifies an abnormal pattern where the edges are not representative samples of the edge population $\textrm{fin}_2(\mathcal{P})$. Thus if we permute the order of these edges in the existing edge sequence, the resulting edge sequence is nonexchangeable.

\vspace{0.2in}
\noindent
{\bf Example 3. University email network \citep{williamson2016nonparametric}:} A university’s email network records email communications among entities in the university. 
The email network follows a consistent communication pattern: administrators may send out a large number of group emails but receive a relatively small number of emails--that is, they often operate in the "sender" role but not the "receiver" role; faculties in collaboration may send and receive emails between each other on a regular basis--they operate as the "sender" role and the "receiver" role both at a moderate frequency. Each email from the email network, $(u, v)$, corresponds to an email from $u$ to $v$. Denote $P_U$ as a distribution of the sender and $P_V$ as a distribution of the receiver.
The emails $(u_1, v_1), \cdots, (u_n, v_n)$ are exchangeable as a sequence in the space of the joint distribution $P_U \times P_V$.

\section{Anomalous Edge Detection} 
\label{ADND-sec:anomaly detection}
In this section, we fit the social network edges using an Asymmetric Dirichlet Network Distribution (ADND) model whose prior is an exchangeable, two-level hierarchical Dirichlet process (HDP). The model accounts for the hierarchical nature of networks, where nodes are grouped into groups that are further divided into subgroups, as well as the asymmetric role that senders and receivers play in forming edges. 
We use a variational inference approach to approximate the posterior distribution of the latent variables and compute the log likelihood of new edges. The anomalous edges with regard to edge exchangeability are then defined. Additionally, based on conformal prediction, we propose Algorithm \ref{alg:base} for detecting the anomalous edges, which has a false positive rate guarantee (Proposition \ref{prop: guarantee}).

\subsection{Asymmetric Dirichlet Network Distribution (ADND) Model}
\label{ADND-subsec: ADND}
Consider a graph $G=(V,E)$, where $V$ is the node set and $E=\{e_1,\cdots,e_N\}$ is the edge set with $N$ edges. We adopt the asymmetric Dirichlet network distribution (ADND) model \citep{williamson2016nonparametric} to fit the network edges. 
Here, each edge is represented as a tuple $e_n=(u_n, v_n), n=1,\cdots,N$, where the nodes $u_n, v_n$ denotes a sender-receiver pair.  

The fundamental concept behind the ADND model involves using Bayesian nonparametrics to define the prior and likelihood for sampling edges from a countably infinite population of nodes. Denote the size of the node population as $W$, the ADND model can be described using a two-level Hierarchical Dirichlet Process (HDP) as follows:
\begin{equation}
\label{ADND-eq:HDP model}
\begin{split}
    &\Theta \sim \textrm{Dir}(\boldsymbol{\eta}), \, \textrm{where } \boldsymbol{\eta} = (\eta, \cdots, \eta) \in \mathbb{R}^{W+1}, \\
    &H \sim \textrm{DP}(\gamma \Theta), \\
    &A \sim \textrm{DP}(\tau H), \quad
    B \sim \textrm{DP}(\tau H), \\
\end{split}
\end{equation}
where $\gamma, \tau > 0$ are the concentration parameters of the first- and second-level Dirichlet process. The first level of the HDP, $H$, represents the topic distribution, where each topic is a distribution over nodes. %
$H$ is itself Dirichlet process distributed -- its base distribution $\Theta$ is a symmetric $(W+1)$-dimensional Dirichlet distribution with concentration parameter $\boldsymbol{\eta}$, where the $(W+1)$-th dimension is reserved for unseen nodes. 
The sender distribution $A$ and the receiver distribution $B$ are topic distributions that share the same set of topics with $H$, but weigh them using specific proportions $\boldsymbol{\beta}^{(A)}$ and $\boldsymbol{\beta}^{(B)}$.

\begin{figure}
	\centering
	\includegraphics[width=0.8\textwidth]{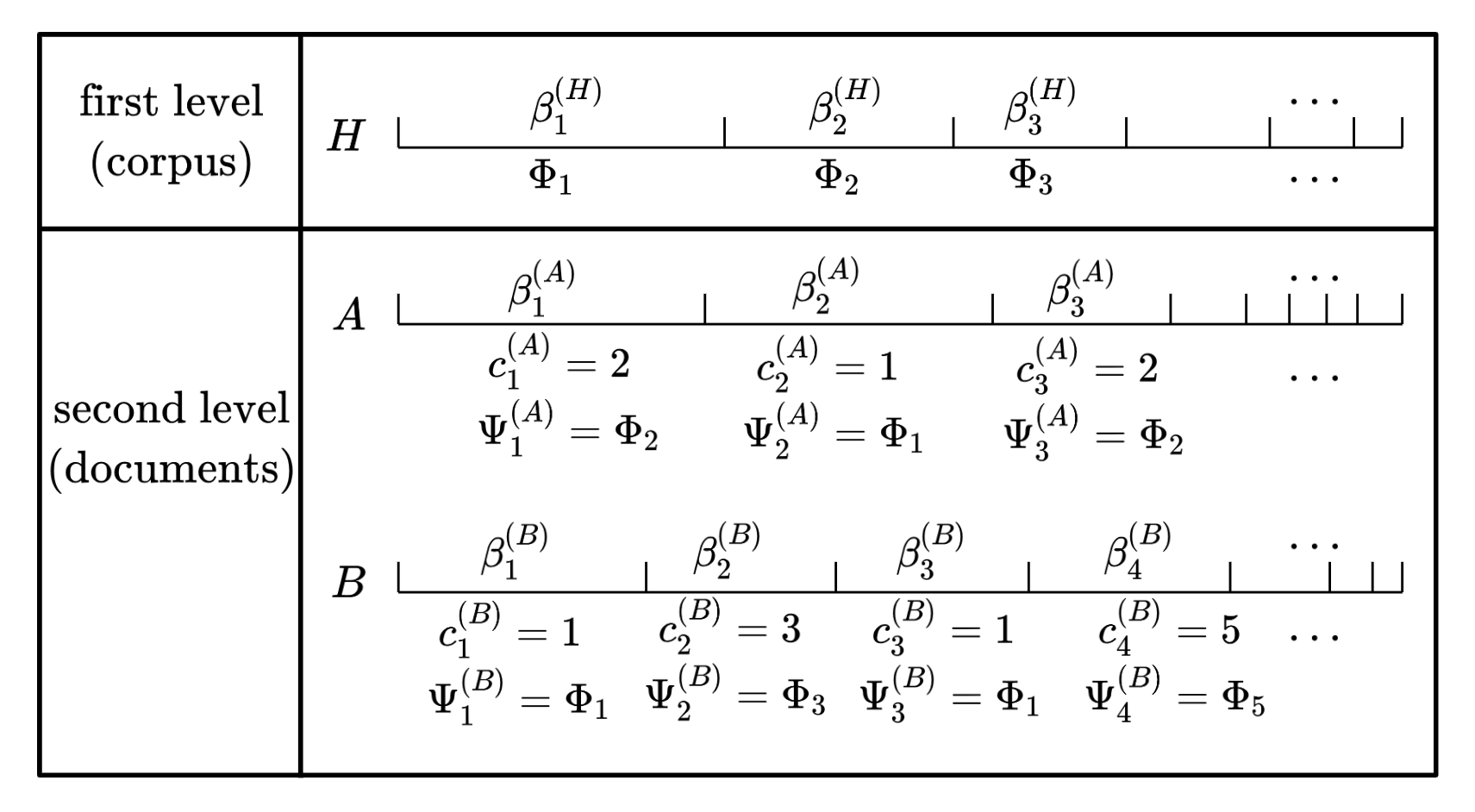}
	\caption{Illustration of the two-level stick-breaking representation of the hierarchical Dirichlet process. In the first level, $\phi_i \sim \Theta$ and $\boldsymbol{\beta}^{(H)} \sim \textrm{GEM}(\gamma)$; in the second level, $\boldsymbol{\beta}^{(A)}, \boldsymbol{\beta}^{(B)} \sim \textrm{GEM}(\tau)$,  $c_{t}^{(A)}, c_{t}^{(B)} \sim \textrm{Categorical}(\boldsymbol{\beta}^{(H)})$ and $\psi_{t}^{(A)} = \phi_{c_{t}^{(A)}}, \psi_{t}^{(B)} = \phi_{c_{t}^{(B)}}$.}
	\label{ADND-fig:stick-breaking}
\end{figure}

We follow \citet{wang2011online}'s stick-breaking representation of the HDP in (\ref{ADND-eq:HDP model}), which enables closed form variational updates in Section \ref{ADND-subsec:VI}. For the corpus-level DP draw, this representation is 
\begin{equation}
\label{ADND-eq:stick-breaking}
\begin{split}
    &\beta_{i}^{(H)'} \sim \textrm{Beta}(1, \gamma), \quad
    \beta_{i}^{(H)}  = \beta_{i}^{(H)'} \prod_{l=1}^{i-1} (1 - \beta_{l}^{(H)'}), \\
    &\phi_i \sim \Theta, \\
    &H = \sum_{i=1}^{\infty} \beta_{i}^{(H)} \delta_{\phi_i}. \\
\end{split}
\end{equation}
Thus, $H$ is discrete and has support at the atoms $\boldsymbol{\phi} = (\phi_i)_{\substack{i=1,\cdots, \infty}}$ with weights $\boldsymbol{\beta}^{(H)} = (\beta^{(H)}_i)_{\substack{i=1,\cdots, \infty}}$. The distribution for $\boldsymbol{\beta}^{(H)}$ is also written as $\boldsymbol{\beta}^{(H)} \sim \textrm{GEM}(\gamma)$ ("GEM" stands for Griffiths, Engen and McCloskey).

To represent the document-level atoms $\boldsymbol{\psi}^{(A)} = (\psi^{(A)}_{t})_{\substack{t=1,\cdots,\infty}}$ and $\boldsymbol{\psi}^{(B)} = (\psi^{(B)}_{t})_{\substack{t=1,\cdots,\infty}}$, we introduce two series of indicator variables, $\mathbf{c}^{(A)} = (c_{t}^{(A)})_{t=1,\cdots, \infty}$ and $\mathbf{c}^{(B)} = (c_{t}^{(B)})_{t=1,\cdots, \infty}$, which are drawn i.i.d.,
\begin{equation}\label{ADND-eq:categorical}
\begin{split}
    & c_{t}^{(A)}, c_{t}^{(B)} \sim \textrm{Categorical}(\boldsymbol{\beta}^{(H)})
\end{split}
\end{equation}
Then let
\begin{equation}
\label{ADND-eq:psi and phi}
\begin{split}
    & \psi_{t}^{(A)} = \phi_{c_{t}^{(A)}}, \quad
     \psi_{t}^{(B)} = \phi_{c_{t}^{(B)}}, \\
\end{split}
\end{equation}
By doing this, we can avoid having to explicitly express the document atoms $\boldsymbol{\psi}^{(A)}$ and $\boldsymbol{\psi}^{(B)}$.

Given the representation in Eq. (\ref{ADND-eq:stick-breaking}, \ref{ADND-eq:categorical}, \ref{ADND-eq:psi and phi}), the generative process for the $n$-th edge, $(u_{n}, v_{n})$, is as follows,
\begin{equation}
\label{ADND-eq:edge generative}
\begin{split}
    &z_{n}^{(A)} \sim \textrm{Categorical}(\boldsymbol{\beta}^{(A)}), \quad
    z_{n}^{(B)} \sim \textrm{Categorical}(\boldsymbol{\beta}^{(B)}), \\
    &u_{n} \sim \textrm{Categorical}(\phi_{c_{z_{n}^{(A)}}^{(A)}}), \quad
    v_{n} \sim \textrm{Categorical}(\phi_{c_{z_{n}^{(B)}}^{(B)}}). \\
\end{split}
\end{equation}
The indicator variables $z_{n}^{(A)}, z_{n}^{(B)}$ select document-level topics $\psi_{t}^{(A)}$ and $\psi_{t}^{(B)}$, respectively, according to the document-level stick weights $\boldsymbol{\beta}^{(A)}$ and $\boldsymbol{\beta}^{(B)}$. $\psi_{t}^{(A)}, \psi_{t}^{(B)}$ map to one of corpus-level topics $\boldsymbol{\phi}$ through the indicators $c_{t}^{(A)}$ and $c_{t}^{(B)}$. We describe the process of drawing senders of edges as follows: for the $n$-th edge, first choose the sender's topic indicator $z_{n}^{(A)} = t$; then choose the corpus-level topic indicator $c_{t}^{(A)} = i$. In this way, $\phi_{i}$ is the sender's chosen topic. We have,
\begin{equation}
    \begin{split}
        u_{n} \bigl\lvert z_{n}^{(A)} = t, c_{t}^{(A)} = i, \boldsymbol{\phi} \sim \textrm{Categorical}(\phi_i)
    \end{split}
\end{equation}
Similarly, for the receiver,
\begin{equation}
    \begin{split}
        v_{n} \bigl\lvert z_{n}^{(B)} = t, c_{t}^{(B)} = i, \boldsymbol{\phi} \sim \textrm{Categorical}(\phi_i)
    \end{split}
\end{equation}

\vspace{0.2in}
\noindent
{\bf Discussion of the ADND model: }
As defined in Section \ref{ADND-subsec: SBM intro}, vertex exchangeable models such as the SBM and IRM assume a fixed node set and model the network by clustering the nodes. The edges are represented as nonzero-valued entries of the corresponding adjacency matrix.
In contrast, ADND models the network as a sequence of edges, and clusters edges into a hierarchy of topics, rather than nodes. Each node can be associated with multiple topics based on its incoming and outgoing edges. Further, by placing a nonparametric prior distribution over the topics, the ADND model can incorporate previously unseen nodes and is flexible for predicting unseen edges.

\subsection{Variational Inference}
\label{ADND-subsec:VI}
In the above representation of ADMD model, the observed data is 
$x=(u_{n}, v_{n})_{\substack{\scriptscriptstyle n=1,\cdots, \infty}}$, the hidden variables are $z=\{(z_{n}^{(A)})_{\substack{\scriptscriptstyle n=1,\cdots, \infty}}, 
(z_{n}^{(B)})_{\substack{ \scriptscriptstyle n=1,\cdots, \infty}}, 
(c_{t}^{(A)})_{\substack{ \scriptscriptstyle t=1,\cdots, \infty}}, 
(c_{t}^{(B)})_{\substack{ \scriptscriptstyle t=1,\cdots, \infty}},\\
(\beta_{t}^{(A)})_{\substack{ \scriptscriptstyle t=1,\cdots, \infty}}, 
(\beta_{t}^{(B)})_{\substack{ \scriptscriptstyle t=1,\cdots, \infty}}, 
(\beta_{i}^{(H)})_{\scriptscriptstyle i=1,\cdots,\infty}, 
(\phi_i)_{\scriptscriptstyle i=1,\cdots,\infty} \}$.
Due to the difficulty in directly computing $p(z|x)$, we use a fully factorized variational distribution for approximation, %
\begin{equation}
\begin{split}
    &q(\boldsymbol{\beta}^{(H)'}, \boldsymbol{\beta}^{(A)'}, \boldsymbol{\beta}^{(B)'}, \boldsymbol{c}^{(A)}, \boldsymbol{c}^{(B)}, \boldsymbol{z}^{(A)}, \boldsymbol{z}^{(B)}, \boldsymbol{\phi}) \\
     = &q(\boldsymbol{\beta}^{(H)'})q(\boldsymbol{\beta}^{(A)'})q(\boldsymbol{\beta}^{(B)'})q(\boldsymbol{c}^{(A)})q(\boldsymbol{c}^{(B)})q(\boldsymbol{z}^{(A)})q(\boldsymbol{z}^{(B)})q(\boldsymbol{\phi}) 
\end{split}
\end{equation}
which further factorizes into
\begin{subequations}
    \begin{align}
        q(\boldsymbol{c}^{(A)}) &= \prod_{t=1}^{K^{(A)}} q(c_{t}^{(A)} | \varphi_{t}^{(A)}), \quad
        q(\boldsymbol{c}^{(B)}) = \prod_{t=1}^{K^{(B)}} q(c_{t}^{(B)} | \varphi_{t}^{(B)}), \\
        q(\boldsymbol{z}^{(A)}) &= \prod_{n} q(z_{n}^{(A)} | \zeta_{n}^{(A)}), \quad
        q(\boldsymbol{z}^{(B)}) = \prod_{n} q(z_{n}^{(B)} | \zeta_{n}^{(B)}), \\
        q(\boldsymbol{\phi}) &= \prod_{i=1}^{K^{(H)}} q(\phi_{i} | \lambda_{i}), 
    \end{align}
\end{subequations}
where the variational parameters are $\varphi_{t}^{(A)}$ (multinomial), $\varphi_{t}^{(B)}$ (multinomial), $\zeta_{n}^{(A)}$ (multinomial), $\zeta_{n}^{(B)}$ (multinomial), and $\lambda_i$ (Dirichlet). We set the corpus-level topic truncation $K^{(H)}$ and the document-level truncation $K^{(A)}$ and $K^{(B)}$. $K^{(H)}$ is usually much bigger than $K^{(A)}$ or $K^{(B)}$ because in practice, the number of topics needed for each document is far fewer than that required for the corpus. The factorized forms of the corpus-level stick proportions and bottom-level stick proportions are 
\begin{subequations}
    \begin{align}
        q(\boldsymbol{\beta}^{(H)'}) &= \prod_{i=1}^{K^{(H)}} q(\beta_{i}^{(H)'} | a^{(H)}_{i}, b^{(H)}_{i}), \\
        q(\boldsymbol{\beta}^{(A)'}) &= \prod_{t=1}^{K^{(A)}} q(\beta_{t}^{(A)'} | a^{(A)}_{t}, b^{(A)}_{t}), \quad
        q(\boldsymbol{\beta}^{(B)'}) = \prod_{t=1}^{K^{(B)}} q(\beta_{t}^{(B)'} | a^{(B)}_{t}, b^{(B)}_{t}), 
    \end{align}
\end{subequations}
where $(a^{(H)}_{i}, b^{(H)}_{i})$, $(a^{(A)}_{t}, b^{(A)}_{t})$, and $(a^{(B)}_{t}, b^{(B)}_{t})$ are parameters of beta distributions. 

The evidence lower bound (ELBO) in variational theory \citep{blei2017variational} lower bounds the marginal log-likelihood of the observed data. ELBO is defined as follows.
\begin{equation}
\begin{split}
    & \mathbb{E}_{q}[\textrm{log}\frac{p(x, z)}{q(z)}]
    = \mathbb{E}_{q}[\textrm{log}(p(x, z)] + H(q)
    \\
    & = 
    \mathbb{E}_{q} 
    \Big[\textrm{log} p(\mathbf{u} | \mathbf{c}^{(A)}, \mathbf{z}^{(A)}, \mathbf{\phi}) p(\mathbf{c}^{(A)} | \boldsymbol{\beta}^{(H)}) p(\mathbf{z}^{(A)} | \boldsymbol{\beta}^{(A)})
    \\
    & p(\boldsymbol{\beta}^{(A)}) p(\mathbf{v} | \mathbf{c}^{(B)}, \mathbf{z}^{(B)}, \boldsymbol{\phi}) p(\mathbf{c}^{(B)} | \boldsymbol{\beta}^{(H)}) p(\mathbf{z}^{(B)} | \boldsymbol{\beta}^{(B)}) p(\boldsymbol{\beta}^{(B)})
    \Big]
    \\
    & + H(q(\mathbf{c}^{(A)})) + H(q(\mathbf{z}^{(A)})) + H(q(\boldsymbol{\beta}^{(A)})) + H(q(\mathbf{c}^{(B)}))
    + H(q(\mathbf{z}^{(B)})) + H(q(\boldsymbol{\beta}^{(B)})) \\
    & + \mathbb{E}_{q} \Big[\textrm{log} p(\boldsymbol{\beta}^{(H)}) p(\boldsymbol{\phi}) \Big] 
    + H(q(\boldsymbol{\beta}^{(H)})) + H(q(\boldsymbol{\phi})),
\end{split}
\end{equation}
where $H(\cdot)$ is the entropy of the variational distribution. Taking derivatives of this lower bound with respect to each variational parameter leads to the coordinate
ascent updates.

\noindent
{\bf Document-level Updates: } At the document level, we update the following parameters,
\begin{subequations}
\label{ADND-eq:document update}
\begin{align}
    &a_{t}^{(A)} = 1 + \sum_{n} \zeta^{(A)}_{nt}, \quad b_{t}^{(A)} = \tau + \sum_{n} \sum_{s=t+1}^{K^{(A)}} \zeta^{(A)}_{ns}, \\
    &a_{t}^{(B)} = 1 + \sum_{n} \zeta^{(B)}_{nt}, \quad b_{t}^{(B)} = \tau + \sum_{n} \sum_{s=t+1}^{K^{(B)}} \zeta^{(B)}_{ns}, \\
    &\varphi_{ti}^{(A)} \propto \textrm{exp}(\sum_{n} \zeta_{nt}^{(A)} \mathbb{E}_q[\textrm{log}p(u_{n}|\phi_i)] + \mathbb{E}_q[\textrm{log}\beta^{(H)}_{i}])  \\
    &\zeta_{nt}^{(A)} \propto \textrm{exp}(\sum_{i=1}^{k^{(H)}} \varphi_{ti}^{(A)} \mathbb{E}_q[\textrm{log}p(u_{n}|\phi_i)] + \mathbb{E}_q[\textrm{log}\beta^{(A)}_{t}])  \\
    &\varphi_{ti}^{(B)} \propto \textrm{exp}(\sum_{n} \zeta_{nt}^{(B)} \mathbb{E}_q[\textrm{log}p(u_{n}|\phi_i)] + \mathbb{E}_q[\textrm{log}\beta^{(H)}_{i}])  \\
    &\zeta_{nt}^{(B)} \propto \textrm{exp}(\sum_{i=1}^{k^{(H)}} \varphi_{ti}^{(B)} \mathbb{E}_q[\textrm{log}p(u_{n}|\phi_i)] + \mathbb{E}_q[\textrm{log}\beta^{(B)}_{t}])  
\end{align}
\end{subequations}

\noindent
{\bf Corpus-level Updates: } At the corpus level, we update the following parameters,
\begin{subequations}
\label{ADND-eq:corpus update}
\begin{align}
    &a_{i}^{(H)} = 1 + [\sum_{t=1}^{K^{(A)}} \phi^{(A)}_{ti} + \sum_{t=1}^{K^{(B)}} \phi^{(B)}_{ti}], \\
    &b_{i}^{(H)} = \gamma + [\sum_{t=1}^{K^{(A)}} \sum_{l=i+1}^{K^{(H)}} \varphi_{tl}^{(A)} + \sum_{t=1}^{K^{(B)}} \sum_{l=i+1}^{K^{(H)}} \varphi_{tl}^{(B)} ], \\
    &\lambda_{iw} = \eta + [\sum_{t=1}^{K^{(A)}} \varphi_{ti}^{(A)}(\sum_{n} \zeta_{nt}^{(A)} \mathbf{1}(u_{n}=w)) +  \sum_{t=1}^{K^{(B)}} \varphi_{ti}^{(B)}(\sum_{n} \zeta_{nt}^{(B)} \mathbf{1}(v_{n}=w))] 
\end{align}
\end{subequations}
The expectations involved are taken under the variational distribution $q$ and can be found in \citep{wang2011online}.

\vspace{0.2in}
\noindent
{\bf Holdout Likelihood of New Edges: } Given the variational inference for the HDP, we can compute the log likelihood of the new edge $(u_{\textrm{new}}, v_{\textrm{new}})$ as follows:
\begin{equation}
\label{ADND-eq:posterior}
\begin{split}
    \textrm{log}P(\small(u_{\textrm{new}}, v_{\textrm{new}} \small)) &= \textrm{log} (\sum_{i=1}^{K^{(H)}} \bar{\beta}_{i}^{(H)} \bar{\lambda}_{i u_{\textrm{new}}} \ \bar{\beta}_{i}^{(H)} \bar{\lambda}_{i v_{\textrm{new}}}) \\
\end{split}
\end{equation}
where $\bar{\boldsymbol{\lambda}} \in \mathbb{R}^{K^{(H)} \times (W+1)}$ is the variational expectation of the corpus-level topic distribution $\boldsymbol{\lambda}$ given the observation $(u_{n}, v_{n})_{\substack{\scriptscriptstyle n=1,\cdots, \infty}}$, and $\bar{\boldsymbol{\beta}}^{(H)} \in \mathbb{R}^{K^{(H)}}$ is the variational expectation of the corpus-level stick lengths measuring the expected proportions of the corpus-level topics.

\subsection{Detecting Anomalous Edges with Conformal Prediction}
\label{ADND-subsec: detect anomaly}
The previous subsections gave a setup of the edge exchangeable model. In this subsection, we define an anomalous edge that is significantly different from the remaining edges which were generated to be edge-exchangeable. We formulate the anomaly detection problem as a hypothesis test of the exchangeability assumption. We then propose an algorithm, a Conformal detector, which has a guaranteed upper bound for the false positive rate.  

\begin{definition}[Anomalous Edges]\label{defn: anomaly}
An anomalous edge $e_{N+1}$ is an edge which breaks the exchangeability of an edge sequence $e_1,\cdots, e_N$ if 
\begin{equation}
  (e_1,\cdots, e_{N}) \overset{d}{=} (e_{\pi_1},\cdots, e_{\pi_{N}}),
\end{equation}
for any permutation $\pi: \{1,\cdots,N\} \rightarrow \{1,\cdots,N\}$,
and
\begin{equation}
  (e_1,\cdots,e_{N+1}) \overset{d}{\neq} (e_{\pi^+_1},\cdots,e_{\pi^+_{N+1}})
\end{equation}
for some permutation $\pi^+: \{1,\cdots,N+1\} \rightarrow \{1,\cdots,N+1\}$.    
\end{definition}

Suppose an existing edge sequence $e_1,e_2,\cdots,e_{N}$ is either constructed or approximated using the ADND model in Section \ref{ADND-subsec: ADND}, then the edge sequence is exchangeable, i.e.,
\begin{equation}
  (e_1,e_2,\cdots,e_{N}) \overset{d}{=} (e_{\pi_1},e_{\pi_2},\cdots,e_{\pi_{N}}),
\end{equation}
According to Definition \ref{defn: anomaly}, an edge $e_{N+1}$ is anomalous if it breaks edge exchangeability, i.e.,
\begin{equation}
  (e_1,e_2,\cdots,e_{N+1}) \overset{d}{\neq} (e_{\pi_1},e_{\pi_2},\cdots,e_{\pi_{N+1}}),
\end{equation}
for some permutation $\pi^+: \{1,\cdots,N+1\} \rightarrow \{1,\cdots,N+1\}$. We next demonstrate how we use such a definition of anomalous edges to construct an anomaly detection algorithm.

\vspace{0.2in}
\noindent
{\bf Conformal Anomalous Edge Detection: }
Our algorithm is based on conformal prediction \citep{shafer2008tutorial}. 
Conformal prediction uses labeled training data to determine the confidence in new predictions. Its output is a prediction set that is guaranteed to contain the true label with a specified confidence level as long as the data is exchangeable \citep{zeni2020conformal}. To this end, conformal prediction estimates a $p$-value for each possible label for a new example. In our application, we have training edges that are labeled as normal edges, and we aim to use conformal prediction to determine the confidence that new edges are also normal. Specifically, the conformal detection algorithm we propose (Algorithm \ref{alg:base}) detects whether an edge $e_{N+1}$ is anomalous or not by using conformal prediction to estimate the $p$-value for the null hypothesis: 

\begin{equation}
\begin{split}
  H_0:  e_1,e_2,\cdots,e_{N+1} \textrm{ forms an exchangeable sequence.}
\end{split}
\end{equation}

In order to estimate the $p$-values for the null hypothesis $H_0$, we employ a \emph{non-conformity measure} (NCM) \citep{shafer2008tutorial} to measure how different an observation is relative to other observations. 
In detecting anomalous edges, to measure how different an edge $e_{n}$ is relative to other edges in the edge set $E=\{e_1,\cdots, e_{N+1}\}$, we define the NCM as the negative likelihood of edge $e_{n}$ given other edges in the edge set $E$, which can be computed according to Eq.~(\ref{ADND-eq:posterior}) as follows:
\begin{equation} \label{ADND-eq: NCM}
\begin{split}
    \alpha_{n} &= -P(e_{n}|E\setminus\{e_{n}\}) \\
\end{split}
\end{equation}
$\alpha_n$ in (\ref{ADND-eq: NCM}) constitutes the test statistic for our detector, which determines the corresponding $p$-value as follows:
\begin{equation} \label{ADND-eq: p-value}
\begin{split}
  p_{n} = &
  \frac{1}{N+1} \bigl \lvert\{i=1,\cdots,N+1 \mid \alpha_{n} > \alpha_{i} \}\bigr \rvert  +
  \frac{u}{N+1} \bigl \lvert\{i=1,\cdots,N+1 \mid \alpha_{n} = \alpha_{i} \}\bigr \rvert, \\
\end{split}
\end{equation}
where $u\sim \textrm{Unif}(0,1)$. Given a significance level or an \emph{anomaly threshold}, $\epsilon$, if $ p_{n} \le \epsilon$, then $e_{n}$ is classified as an anomalous edge: $\textrm{Anom}_{n}^{\epsilon} = 1$. Otherwise, $e_{n}$ is classified as normal: $\textrm{Anom}_{n}^{\epsilon} = 0$. 

Note that we will need to fit the ADND model to compute the log likelihood in Eq.~(\ref{ADND-eq:posterior}) for any $E\setminus \{e_n\}, n=1,\cdots,N+1$. To avoid such computation burden, we take advantage of the inductive conformal predictor \citep{papadopoulos2008inductive}. We split\footnote{Due to the edge-exchangeability, the split can be done randomly. } the existing edge set (excluding the new edge to be tested) into a training set $E_{\textrm{train}}$ and a calibration set $E_{\textrm{calib}}$. We fit the ADND model based on $E_{\textrm{train}}$ and compute the non-conformity scores for each observation in $E_{\textrm{calib}}$. The $p$-value of the new edge is computed by comparing its non-conformity score with the ones of $E_{\textrm{calib}}$:
\begin{equation} \label{ADND-eq: inductive p-value}
\begin{split}
  p_{n} = &
  \frac{\bigl \lvert\{i\in E_{\textrm{calib}} \cup \{N+1\} \mid \alpha_{n} > \alpha_{i} \}\bigr \rvert }{|E_{\textrm{calib}}|+1}  +
  \frac{u \bigl \lvert\{i\in E_{\textrm{calib}} \cup \{N+1\} \mid \alpha_{n} = \alpha_{i} \}\bigr \rvert}{|E_{\textrm{calib}}|+1} , \\
\end{split}
\end{equation}

We present the complete anomaly detection algorithm (Conformal detector) in Algorithm \ref{alg:base}. Presented below is our main finding: we demonstrate that Algorithm \ref{alg:base} has a false positive rate (Definition \ref{defn: fpr}), meaning that it incorrectly detects normal edges as anomalies, that is no greater than the predetermined anomaly threshold of $\epsilon$.

\begin{definition}[False Positive Rate]\label{defn: fpr}
A false positive rate refers to the proportion of  non-anomalous data points that are incorrectly labeled as anomalous or abnormal. For the conformal anomaly detection algorithm, the false positive rate represents the proportion of normal edges that are incorrectly detected as anomalous (Definition \ref{defn: anomaly}).  
\end{definition}

\begin{proposition}\label{prop: guarantee}
Assuming that a new edge $e_{N+1}$ is normal in relation to the existing edges, i.e., the edges $e_1,e_2,\cdots,e_{N+1}$ form an exchangeable sequence. For the choice of anomaly score in Eq. (\ref{ADND-eq: NCM}), the specified anomaly threshold $\epsilon$ corresponds to an upper bound of the false positive rate (Definition \ref{defn: fpr})
of Algorithm 1:
\begin{equation}
  P(\textrm{Anom}_{N+1}^{\epsilon} =1) \leq \epsilon
\end{equation}
\end{proposition}
\begin{proof}
    See Appendix.
\end{proof}

\begin{algorithm2e}\label{alg:base}
\caption{(Conformal Detector) Anomalous edge detection based on inductive conformal prediction}
\label{alg:net}
\DontPrintSemicolon
\LinesNumbered
\KwIn{anomaly threshold $\epsilon$, training edge set $E_{\textrm{train}}$, calibration edge set $E_{\textrm{calib}}$, and a new edge $e_{N+1}$ to be tested.}
\KwOut{Indicator variable $\textrm{Anom}_{N+1}^{\epsilon} \in \{0,1\}$ with performance guarantee (Proposition \ref{prop: guarantee}).}

Initialize the corpus-level topic distributions $\boldsymbol{\lambda} = (\boldsymbol{\lambda}_i)_{\scriptscriptstyle i=1, \cdots, K^{(H)}}$, the parameters of the beta distribution $\boldsymbol{a}^{(H)} = (a^{(H)}_{i})_{\scriptscriptstyle i=1, \cdots, K^{(H)}-1}$, $\boldsymbol{b}^{(H)} = (b^{(H)}_{i})_{\scriptscriptstyle i=1, \cdots, K^{(H)}-1}$.

\tcc{\textcolor{blue}{Approximate the training edges using the ADND edge distribution}}
\While{stopping criterion is not met}
    {Compute $\boldsymbol{a}^{(A)}, \boldsymbol{b}^{(A)}, \boldsymbol{a}^{(B)}, \boldsymbol{b}^{(B)}, \boldsymbol{\psi}_{t}^{(A)}, \boldsymbol{\zeta}_{n}^{(A)}, \boldsymbol{\psi}_{t}^{(B)}$ and $\boldsymbol{\zeta}_n^{(B)}$ using variational inference through document-level updates, Eq. (\ref{ADND-eq:document update}).

    Update $\boldsymbol{\lambda}, \boldsymbol{a}^{(H)}, \boldsymbol{b}^{(H)}$ using Eq. (\ref{ADND-eq:corpus update}).}

\tcc{\textcolor{blue}{Compute the non-conformity scores of the calibration edges}}
\For{$j \leftarrow 1 \textrm{ to } |E_{\textrm{calib}}| $} 
{Compute $\alpha_j = - P(e_j|E_{\textrm{train}})$ according to Eq.~(\ref{ADND-eq:posterior}).}

\tcc{\textcolor{blue}{Compute the $p$-value that the test edge is normal}}

Compute $p_{N+1}$ according to Eq.~(\ref{ADND-eq: inductive p-value}).

\uIf{$p_{N+1} \le \epsilon$}{
$\textrm{Anom}_{N+1}^{\epsilon} \leftarrow 1$}       
\Else{
$\textrm{Anom}_{N+1}^{\epsilon} \leftarrow 0$}

\end{algorithm2e}

\section{Numerical Examples of Anomalous Edge Detection}
\label{ADND-sec: numerical}
In this section, we illustrate how the proposed Conformal detector (Algorithm 1) can detect anomalous transactions in the Ethereum blockchain transaction network.

\vspace{0.2in}
\noindent
{\bf Dataset.}
Decentralized exchanges (DEXs), which enable the non-custodial trade of digital assets, are implemented on the Ethereum blockchain. The most popular DEX that uses an automated liquidity protocol is Uniswap \citep{adams2021uniswap}. Since Uniswap V2, it has allowed the creation of arbitrary ERC20/ERC20\footnote{The ERC20 (\textit{Ethereum Request for Comment \#20}) stands for fungible tokens on the Ethereum blockchain.} pairs that store the pooled reserves of two different assets.
Uniswap V3, the most recent version, debuts in May 2021. Compared to top centralized exchanges (CEXs), Uniswap V3 is more liquid \citep{UniswapV3}. 

In Uniswap, traders can exchange assets using a swap function that involves adding one asset to a liquidity pair while removing the other. To improve routing efficiency, traders can perform multi-hop swaps, such as exchanging token $X$ for $Y$ and then exchanging $Y$ for $Z$. However, if a liquidity pair for tokens $X$ and $Z$ does not exist, this can be treated as a potential pricing anomaly and an arbitrage opportunity. In our experiment, we use this scenario as the ground truth anomaly to be detected.

To do this, we gathered all the multi-hop swap transactions in Uniswap V3. We accessed the Ethereum Mainnet using the Python library \textit{Web3.py}\footnote{\url{https://web3py.readthedocs.io/}}. 
We searched all blocks with a height between $12500000$ and $12503000$, which corresponds to about three weeks after the debut of Uniswap V3 on May 5, 2021. A directed edge from the source token to the destination token is used to model each transaction. There are 140 anomalous edges in our collection of 307 test edges, 363 training edges, and 363 calibration edges.

\vspace{0.2in}
\noindent
{\bf Baseline.} 
We use the RHSS method \citep{ranshous2015anomaly} as the baseline, which is an edge stream anomaly detector applicable to directed multigraphs. RHSS estimates the edge probabilities by combining three individual scores, i.e., sample, preferential attachment, and homophily scores. We assign equal weighting coefficients to the individual scores.

\vspace{0.2in}
\noindent
{\bf Evaluation metrics.}
Both the Conformal detector and the RHSS method output a score per edge (the Conformal detector outputs a $p$-value; the RHSS method outputs an estimate of edge probabilities. Both are lower for more anomalous edges). We compute the precision and recall following the method outlined in \citep{eswaran2018sedanspot}: sorting the edges in ascending order of their scores, we count the number of edges $c_k$ flagged correctly as anomalous among the top $k$ edges, for every cut-off rank $k\in \mathbb{N}$. Let $C$ be the total number of ground truth anomalies, we compute: precision$@k$ = $c_k/k$ and recall$@k$ = $c_k/C$. In addition, we compute the true positive rate (proportion of anomalous edges correctly classified as anomalies) and the false positive rate (proportion of normal edges mistakenly classified as anomalies) for a range of threshold values. We also report AUC (area Under the ROC curve).

\begin{figure} 
	\centering
	\includegraphics[width=0.8\textwidth]{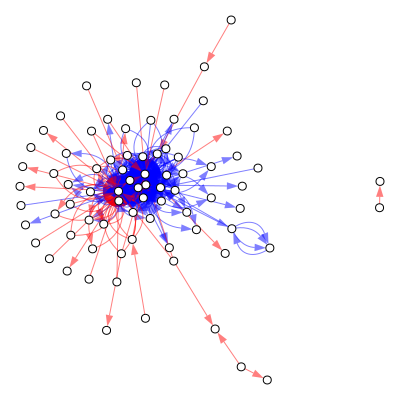}
	\caption{The Uniswap V3 token swap network collected as described in Section \ref{ADND-sec: numerical}. The edges that are normal are blue, whereas the anomalous ones are red. The normal edges are approximated using an exchangeable ADND model, and we aim to detect the anomalous edges that represent tokens swaps without existing liquidity pools. The normal edges and anomalous edges suggest different generating properties.}
	\label{ADND-fig: network}
\end{figure}

Fig. \ref{ADND-fig: network} shows the token swaps collected from Uniswap V3. The normal edges are blue; the anomalous edges are red. 
The left of Fig. \ref{ADND-fig: PR} shows the precision vs. recall curves, while the right shows the ROC curve and the AUC computed. The performance of the Conformal detector is comparable to that of the RHSS method. 
Overall, the comparison result indicates that the Conformal detector is comparable with the RHSS method in anomalous edge detection. 
Moreover, the relatively small performance gap illustrates the robustness of the Conformal detector to misspecified models. 

\begin{figure} 
	\centering
	\includegraphics[width=0.8\textwidth]{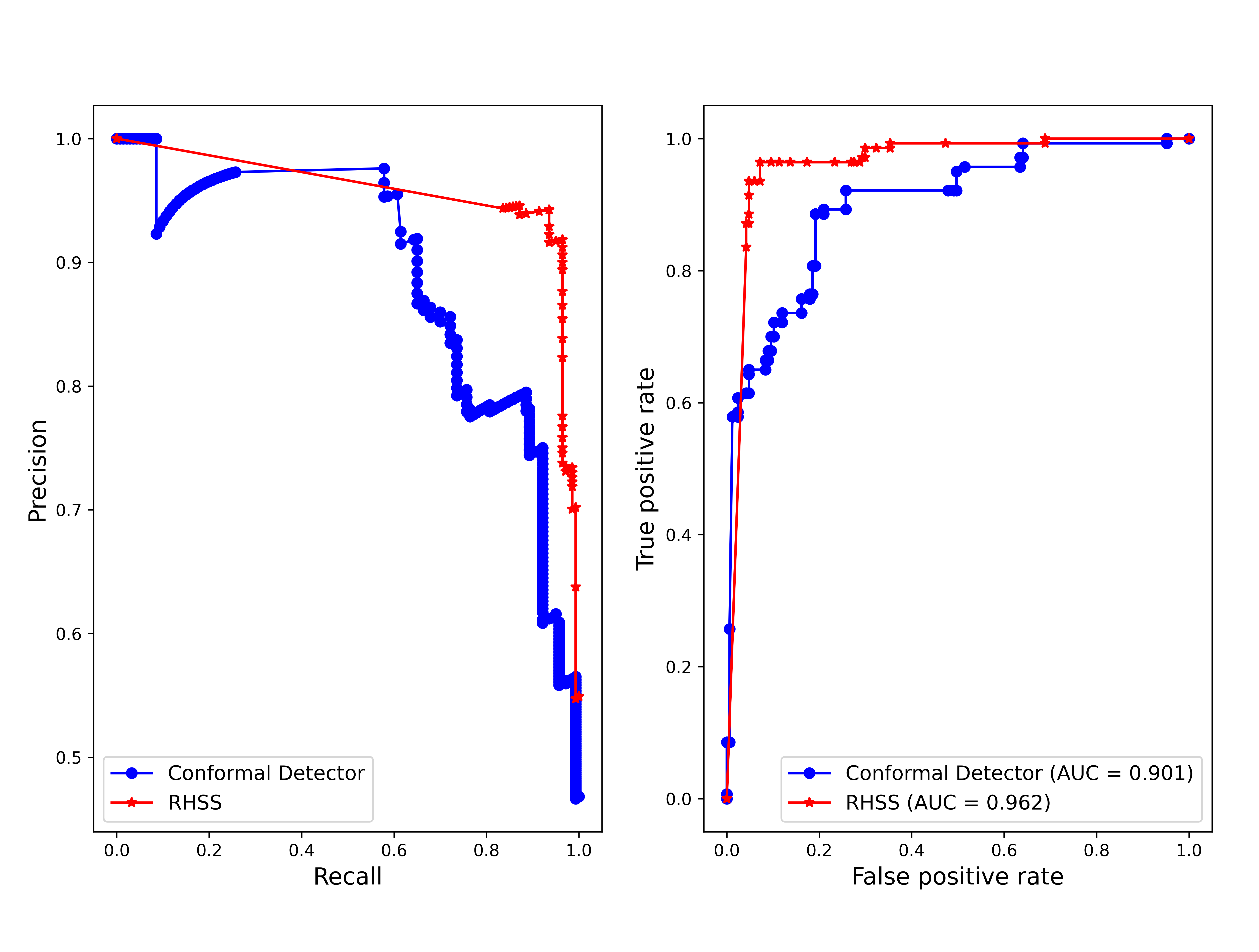}
	\caption{The precision-recall curve and the ROC curve of the Conformal Detector (blue circles), and the RHSS method (red stars). The Conformal Detector is comparable to the RHSS method.}
	\label{ADND-fig: PR}
\end{figure}

\section{Conclusions and Extensions}
\label{ADND-sec:conclusion}
\noindent
{\bf Conclusions: }
This paper addressed detecting anomalous edges in social networks and transaction graphs based on edge exchangeability. We motivated edge exchangeability in social networks via real-world examples. 
We proposed anomalous edge detection algorithm (Conformal detector) by combining conformal prediction with the edge exchangeable model. We proved that Conformal detector guarantees a false positive rate (flagging normal edges as anomaly) lower than a pre-specified threshold. In the numerical experiments, we verified that edge exchangeability is a valid criterion for distinguishing anomalous edges from normal edges using Ethereum blockchain transaction dataset. In addition, we demonstrated that the proposed Conformal detector achieves a comparable anomaly detection performance to the baseline method on synthetic datasets.

\vspace{0.2in}
\noindent
{\bf Limitations and Extensions: }
There are several limitations of the proposed method that open up future research directions. Firstly, the assumption of edge exchangeability may not hold for real-world networks whose properties change over time as they expand. \citep{ghalebi2019nonparametric, ghalebi2019sequential} have incorporated distance-dependent Chinese restaurant process to model nonexchangeable edge sequence by preferentially assigning edges to communities that have been recently active. Secondly, although the validity of conformal prediction (e.g. false positive rate in the proposed Conformal detector algorithm) is guaranteed, the efficiency of conformal prediction (i.e., the prediction region of normal edges) depends on the non-conformity measure. It is worthwhile to investigate cascades of predictors to improve the efficiency \citep{fisch2020efficient}. Finally, combining the proposed conformal prediction framework with graph neural networks or embedding methods remains an interesting future direction. However, conformal prediction crucially depends on the assumption of exchangeability which requires verification of exchangeability in the graph-structured data.

\acks{We would like to acknowledge support for this project from the U. S. Army Research Office under grant W911NF-21-1-0093, the National Science Foundation under grant CCF-2112457, and a grant  from City University of Hong Kong (Project No. 9610639). }

\newpage
\appendix
\section*{Proof of Proposition \ref{prop: guarantee}}
\label{ADND-sec:appendix}
\begin{definition}[Rank, ascending]
For a set of $n$ real numbers $X=\{x_1,\cdots,x_n\}$, define the rank of $x_i$ among $X$ as 
\begin{equation}
\textrm{rank}(x_i;X) =
\begin{cases}
    & |\{j\in \{1,\cdots,n\}: x_j \leq x_i \}| \\
    &\quad \quad  \quad \textrm{if all elements of $X$ are distinct;} \\
    & |\{j\in\{1,\cdots,n\}:x_j+\xi U_j \leq x_i +\xi U_i  \}| \\
    &\quad \quad  \quad \textrm{if some elements of $X$ are equal.}
\end{cases}
\end{equation}
where $\xi >0$ is a small number such that
\begin{equation}
    \xi \leq \min_{\forall i,j\in \{1,\cdots,n\}: x_i\neq x_j} \frac{|x_i - x_j|}{2}
\end{equation}
and $U_1,\cdots,U_n \sim \textrm{Unif}(-1,1)$ are iid uniform random variables. Because $U_1,\cdots,U_n$ are almost surely distinct, $x_i+\xi U_i$ are also distinct. The addition of $U_1,\cdots,U_n$ breaks potential ties among the original sequence $x_1,\cdots,x_n$.    
\end{definition}

\noindent
{\bf Theorem 1. \citep{kuchibhotla2020exchangeability}}
If $W_1,\cdots,W_n$ are exchangeable random variables, then
\begin{equation}
  P(\textrm{rank}(W_i;\{W_1,\cdots,W_n\}) \leq m) = \frac{\lfloor m \rfloor}{n},
\end{equation}
where $\textrm{rank}(W_i;\{W_1,\cdots,W_n\})$ represents the rank of $W_i$ among $\{W_1,\cdots,W_n\}$, $\lfloor m \rfloor$ represents the largest integer smaller than or equal to $m$. Moreover, the random variable $\textrm{rank}(W_i;\{W_1,\cdots,W_n\})/n$ is a valid $p$-value, i.e.,
\begin{equation}
    P(\textrm{rank}(W_i;\{W_1,\cdots,W_n\})/n \leq \alpha) \leq \alpha \quad \forall \alpha \in [0,1].
\end{equation}

\vspace{0.2in}
\noindent
{\bf Theorem 2. \citep{kuchibhotla2020exchangeability}} 
Suppose $W=(W_1,\cdots,W_n)\in \mathcal{W}^n$ is a vector of exchangeable random variables. Fix a transformation $G:\mathcal{W}^n \rightarrow (\mathcal{W}')^m$. If for each permutation $\bar{\pi}:\{1,\cdots,m\} \rightarrow \{1,\cdots,m\}$ there exists a permutation $\underline{\pi}:\{1,\cdots,n\} \rightarrow \{1,\cdots,n\}$ such that
\begin{equation}
  \bar{\pi} G(w) = G(\underline{\pi} w), \quad \forall w\in \mathcal{W}^n,
\end{equation}
then $G(.)$ preserves exchangeability of $W$.

\vspace{0.2in}
\noindent
{\bf Proposition \ref{prop: guarantee}. (Section \ref{ADND-subsec: detect anomaly})} Assume that $e_1,e_2,\cdots,e_{N+1}$ form a sequence of exchangeable random variables. 
For the NCM as defined in Eq. (\ref{ADND-eq: NCM})), the specified anomaly threshold $\epsilon$ corresponds to an upper bound of the probability of the event that $e_{N+1}$ is classified as an anomalous edge by Algorithm 1:
\begin{equation}
  P(\textrm{Anom}_{N+1}^{\epsilon} =1) \leq \epsilon
\end{equation}
{\bf Proof.} 
Let $e_1,e_2,\cdots,e_{N+1} \in \mathcal{E}$ where $\mathcal{E}$ represents the set of all possible edges.
For any permutation $\pi:\{1,\cdots,N+1\} \rightarrow \{1,\cdots,N+1\}$, 
\begin{equation}
\begin{split}
    &\pi \Bigg(-P(e_1|E_{N+1}\setminus\{e_1\}), \cdots, -P(e_{N+1}|E_{N+1}\setminus\{e_{N+1}\}) \Bigg)
    \\ = &
    \Bigg(-P(e_{\pi_1}|E_{N+1}\setminus\{e_{\pi_1}\}), \cdots, -P(e_{\pi_{N+1}}|E_{N+1}\setminus\{e_{\pi_{N+1}}\}) \Bigg) \\
\end{split}
\end{equation}
Theorem 2 indicates that the negative likelihood of an edge, $\alpha: \mathcal{E}\rightarrow \mathbb{R}$, is a transformation that preserves exchangeability.
Thus $\alpha_1,\cdots,\alpha_{N+1}$ are  real-valued  exchangeable  random  variables.

The right hand side of Eq.~(\ref{ADND-eq: p-value}) is the rank of $\alpha_{N+1}$ among $\{\alpha_1,\cdots,\alpha_{N+1}\}$. 
According to Theorem 2, the $p$-value constructed from Eq.~(\ref{ADND-eq: p-value}) is valid, i.e., $p_{N+1}$ is distributed uniformly on $[0,1]$ when $e_{1},\cdots,e_{N+1} \in \mathcal{E}$ be exchangeable random variables. This implies that $P(\textrm{Anom}_{N+1}^{\epsilon} =1 \mid e_{N+1} \textrm{ is normal})$, i.e., the false alarm rate of Algorithm 1, is less than or equal to $\epsilon$ at any specified anomaly threshold $\epsilon$.


\end{document}